\documentclass[a4paper]{scrartcl}

\usepackage[utf8]{inputenc}
\usepackage[T1]{fontenc}
\usepackage{microtype}
\usepackage{lmodern}

\usepackage{graphicx}
\usepackage{xcolor}
\usepackage[fleqn]{amsmath}
\usepackage{amsfonts}
\usepackage{amssymb}
\usepackage{amsthm}
\usepackage{mathtools}
\usepackage{stmaryrd}
\usepackage{bm}
\usepackage{colortbl}
\usepackage{multirow}
\usepackage{import}

\usepackage{tabularx}
\usepackage[strings]{underscore}

\usepackage{authblk}
\usepackage{thm-restate}

\newtheorem{theorem}   {Theorem}[section]
\newtheorem{definition}{Definition}

\newtheorem{lemma}       [theorem]{Lemma}
\newtheorem{corollary}   [theorem]{Corollary}

\theoremstyle{definition}
\newtheorem{example}     [theorem]{Example}

\definecolor{niceredbright}{HTML}{bd0310}
\definecolor{nicebluebright}{HTML}{197b9b}
\definecolor{nicered}{HTML}{7f0a13}
\definecolor{niceblue}{HTML}{104354}
\definecolor{nicegreen}{HTML}{217516}
\definecolor{nicepurple}{HTML}{884bab}
\definecolor{nicebg}{HTML}{f6f0e4}
\definecolor{niceredlight}{HTML}{c9888d}
\definecolor{nicebluelight}{HTML}{78a4b8}
\definecolor{nicegreenlight}{HTML}{76de68}
\definecolor{nicepurplelight}{HTML}{bc87db}

\makeatletter
\RequirePackage[bookmarks,unicode,colorlinks=true]{hyperref}%
   \def\@citecolor{niceblue}%
   \def\@urlcolor{niceblue}%
   \def\@linkcolor{nicered}%

\def\orcidID#1{\smash{\href{http://orcid.org/#1}{\protect\raisebox{-1.25pt}{\protect\includegraphics{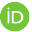}}}}}
\makeatother

\usepackage{etoolbox}
\makeatletter
\pretocmd\start@gather{%
    \if@minipage\kern-\topskip\kern-\baselineskip\kern+7pt\fi
}{}{}
\makeatother

\usepackage{bm}       
\usepackage{stmaryrd} 

\newcommand{\defeq}{\stackrel{\scriptscriptstyle\text{def}}{=}}

\renewcommand{\ldots}{...}

\newcommand{\N}{\mathbb{N}}                         
\def\Z{\mathbb{Z}}                         
\renewcommand{\vec}[1]{#1}                          
\newcommand{\norm}[1]{\lVert#1\rVert}               

\newcommand{\multiset}[1]{\Lbag#1\Rbag}             
\newcommand{\mplus}{\mathbin{+}}                    
\newcommand{\supp}[1]{\llbracket#1\rrbracket}       

\newcommand{\PP}{\mathcal{P}}                       

\renewcommand{\O}{\mathcal{O}}                      

\newcommand{\Tra}{T} 
\newcommand{\BB}{\mathit{BB}} 
\newcommand{\BBL}{\mathit{BB}_L} 
\newcommand{\SC}{\mathit{SC}} 
\newcommand{\InC}[1]{\mathit{IC}({#1})} 
\newcommand{\BBP}{\textbf{BBP}} 
\newcommand{\All}{\mathit{All}} 
\newcommand{\Prot}{\textbf{PP}}
\newcommand{\STC}{\mathit{STATE}} 
\renewcommand{\defeq}{:=}
\newcommand{\Conf}{C}
\newcommand{\Confset}{{\cal C}} 
\newcommand{\Conff}{C'}

\newtheorem{remark}{Remark}

\definecolor{nicered}{HTML}{7f0a13}
\definecolor{nicepurple}{HTML}{884bab}
\newcommand{\trans}[1]{\mathrel{\raisebox{-1pt}[10pt][0pt]{$\xrightarrow{#1}$}}}
\newcommand{\transS}{\mathrel{\raisebox{0pt}[6pt][0pt]{$\xrightarrow{\raisebox{-1pt}[0pt][-1pt]{$\scriptstyle *$}}$}}}
\newcommand{\potrans}[1]{\mathrel{\raisebox{-1pt}[10pt][0pt]{$\xRightarrow{#1}$}}}
\newcommand{\potransS}{\mathrel{\raisebox{0pt}[6pt][0pt]{$\xRightarrow{\raisebox{-1pt}[0pt][-1pt]{$\scriptstyle *$}}$}}}
\newcommand{\Cardinality}[1]{\mathopen\vert#1\mathclose\vert}
\newcommand{\size}[1]{\Cardinality{#1}}
\newcommand{\abs}[1]{\Cardinality{#1}}

\newcommand{\Abs}[1]{\Cardinality{#1}}
\newcommand{\Basel}{B}
\newcommand{\Basecount}{\vartheta}
\newcommand{\LargeC}{D}
\newcommand{\Confa}{D_a}
\newcommand{\Confb}{D_b}
\newcommand{\setN}{\N}

\newcommand{\Displ}[1]{\Delta_{#1}}
\newcommand{\pottcons}{2 (2 \size{T} + 1)^{\size{Q}}}
\newcommand{\pottname}{\xi}

\newcommand{\Pottbase}{{\cal B}}

\usepackage[shortlabels]{enumitem}
\setlist[enumerate,1]{itemsep=0pt,topsep=1ex,before={\pagebreak[1]}}
\setlist[itemize,1]{itemsep=0pt,topsep=1ex}

\newenvironment{appendices}{\appendix}{}

\begin{document}

\title{Lower Bounds on the State Complexity of Population Protocols \footnote{This work was supported by an ERC Advanced Grant (787367: PaVeS), by the Research Training Network of the Deutsche Forschungsgemeinschaft (DFG) (378803395: ConVeY), and by the grant ANR-17-CE40-0028 of the French National Research Agency ANR (project BRAVAS).} \footnote{A previous version of this paper appeared in the proceedings of PODC~2021~\cite{CE21}. The full version of that paper can be found as \texttt{v2} at~\cite{czerner2021lower}}}

\author{Philipp Czerner$^1$ \orcidID{0000-0002-1786-9592}, Javier Esparza$^1$ \orcidID{0000-0001-9862-4919}, Jérôme Leroux$^2$ \orcidID{0000-0002-7214-9467}}
\affil{\{czerner, esparza\}@in.tum.de, jerome.leroux@labri.fr\\
$^1$ Department of Informatics, TU München, Germany\\
$^2$ LaBRI, CNRS, Univ. Bordeaux, France}
\maketitle              

\begin{abstract}
Population protocols are a model of computation in which an arbitrary number of
indistinguishable finite-state agents interact in pairs. The goal of the agents is to decide by stable consensus 
whether their initial global configuration satisfies a given property, specified as a predicate on the set of configurations. The state complexity of a predicate is the number of states of a smallest protocol that computes it. Previous work by Blondin \textit{et al.} has shown that the counting predicates $x \geq \eta$ have state complexity $\O(\log \eta)$ for leaderless protocols and $\O(\log \log \eta)$ for protocols with leaders. We obtain the first non-trivial lower bounds: the state complexity of $x \geq \eta$ is $\Omega(\log\log \eta)$ for leaderless protocols, and the inverse of a non-elementary function for protocols with leaders.

\end{abstract}


\section{Introduction} \label{sec:intro}
Population protocols are a model of computation in which an arbitrary number of
indistinguishable finite-state agents interact in pairs to decide if their initial global configuration satisfies
a given property. Population protocols were introduced in \cite{AngluinADFP04,AngluinADFP06}  to study the theoretical properties networks of mobile sensors with very limited computational resources, but they are also very strongly related to chemical reaction networks, a discrete model of chemistry in which agents are molecules that change their states due to collisions.

Population protocols decide a property by \emph{stable consensus}. Each state of an agent is assigned a binary output (yes/no). In a correct protocol, all agents eventually reach the set of states whose output is the correct answer to the question ``did our initial configuration satisfy the property?'', and stay in it forever. An example of a property decidable by population protocols is majority: initially agents are in one of two initial states, say $A$ and $B$, and the property to be decided is whether the number of agents in $A$ is larger than the number of agents in $B$. In a seminal paper, Angluin \textit{et al.} showed that population protocols can decide exactly the properties expressible in Presburger arithmetic, the first-order theory of addition \cite{AngluinAER07}. 

Assume that at each step a pair of agents is selected uniformly at random and allowed to interact. The \emph{parallel runtime} is defined as the expected number of interactions until a stable consensus is reached (i.e.\ until the property is decided), divided by the number of agents. Even though the parallel runtime is computed using a discrete model, under reasonable and commonly accepted assumptions the result coincides with the runtime of a continuous-time stochastic model. Many papers have investigated the parallel runtime of population protocols, and several landmark results have been obtained. In \cite{AngluinADFP06} it was shown that every Presburger property can be decided in $\O(n\log n)$ parallel time, where $n$ is the number of agents, and \cite{AngluinAE08a} showed that  population protocols with a fixed number of leaders can compute all Presburger predicates in polylogarithmic parallel time. (Loosely speaking, leaders are  auxiliary agents that do not form part  of the population of ``normal'' agents, but can interact with them to help them decide the property.)  More recent results have studied protocols for majority and leader election in which the number of states grows with the number of agents, and shown that polylogarithmic time is achievable by protocols without leaders, even for very slow growth functions, see e.g.~\cite{AAEGR17,AlistarhAG18,AlistarhG18,ElsasserR18,GS20}.

\newcommand{\Flen}{l}
However, many protocols have a high number of states. For example, a quick estimate shows that the fast protocol for majority 
implicitly described in \cite{AngluinAE08a} has tens of thousands of states. This is an obstacle to implementations of protocols in chemistry,
where the number of states corresponds to the number of chemical species participating in the reactions. The number of states is also important because it plays the role of memory in sequential computational models; indeed, the total memory available to a protocol is the logarithm of the number of states multiplied by the number of agents. Despite these facts, the \emph{state complexity} of a Presburger property, defined as the minimal number of states of any protocol deciding the property, has received comparatively little attention\footnote{Notice that the time-space trade-off results of \cite{AAEGR17,AlistarhAG18,AlistarhG18,ElsasserR18,GS20} refer to a more general model in which the number of states of a protocol grows with the number $n$ of agents; in other words, a property is decided by a \emph{family} of protocols, one for each value of $n$. Trade-off results bound the growth rate needed to compute a  predicate within a given time. We study the minimal number of states of a \emph{single} protocol that decides the property for \emph{all} $n$.}. In \cite{BlondinEJ18,BlondinEGHJ20} Blondin \textit{et al.} have shown that every predicate representable by a boolean combination of threshold and modulo constraints (every Presburger formula can be put into this form), with numbers encoded in binary, can be decided by a protocol with polynomially many states in the length of the formula. In particular, it is not difficult to see that every property of the form $x \geq \eta$, stating that the number of agents is at least~$\eta$, can be decided by a leaderless protocol with $\O(\log \eta)$ states. A theorem of \cite{BlondinEGHJ20} also proves the existence of an infinite family of thresholds $\eta$ such that $x \geq \eta$ can be decided by a protocol (with leaders) having $\O(\log\log \eta)$ states. However, to the best of our knowledge there exist no \emph{lower} bounds on the state complexity, i.e.\ bounds showing that a protocol for $x \geq \eta$ needs $\Omega(f(\eta))$ states for some function $f$. This question, which was left open in \cite{BlondinEJ18}, is notoriously hard due to its relation to fundamental questions in the theory of Vector Addition Systems.

In this paper we first show that every protocol, with or without leaders, needs a number of states that, roughly speaking, grows like the inverse Ackermann function, and then prove our main result: every leaderless protocol for $x \geq \eta$ needs $\Omega(\log\log \eta)$ states. The proof of the first bound relies on results on the maximal length of controlled antichains of $\N^d$, a topic in combinatorics with a long tradition in the study of Vector Addition Systems and other models, see e.g. \cite{McAloon84,FigueiraFSS11,AbriolaFS15,Schmitz16,Balasubramanian20}. The double logarithmic bound follows from Pottier's small basis theorem, a useful result of the theory of Diophantine equations \cite{Pottier91}.

The paper is organised as follows. Section~\ref{sec:pre} introduces population protocols, the state complexity function, and its inverse, the busy beaver function, which assigns to a number of states $n$ the largest $\eta$ such that a protocol with $n$ states decides $x \geq \eta$. Instead of lower bounds on state complexity, we present upper bounds on the busy beaver function for convenience. Section~\ref{ssec:stablesets} presents some results on the mathematical structure of stable sets of configurations that are used throughout the paper. Section~\ref{sec:generalbound} shows an Ackermannian upper bound on the busy beaver function, valid for protocols with or without leaders, and explains why this very large bound might be optimal. Section \ref{sec:leaderless bound} gives a triple exponential upper bound on the busy beaver function for leaderless protocols.

\section{Population Protocols and State Complexity} \label{sec:pre}


\subsection{Mathematical preliminaries}
For sets $A,B$ we write $A^B$ to denote the set of functions $f \colon B \rightarrow A$. If $B$ is finite we call the 
elements of $\N^B$ \emph{multisets} over $B$. 
We sometimes write multisets using set-like notation, e.g. $\multiset{a, b, b}$ and $\multiset{a, 2 \cdot b}$
denote the multiset $m$ such that $m(a) = 1$, $m(b) = 2$ and $m(c) = 0$ for every $c \in B \setminus \{a,b\}$.
Given a multiset $m \in \N^B$ and $B' \subseteq B$, we define $m(B') := \sum_{b \in B'} m(b)$. The \emph{size} of $m$ is $\size{m}:= m(B)$; in other words, the total number of elements of $m$. The \emph{support} of $m$ is the set $\supp{m} = \{b \in B \mid mb) > 0\}$.  Abusing language we identify an element $b\in B$ with the one-element multiset containing it, i.e.\ with the multiset $m \in \N^B$ given by $m(b)=1$ and $m(b')=0$ for $b' \ne b$. 

We call the elements of $\Z^B$ \emph{vectors} over $B$ of dimension $\size{B}$. Observe that every multiset is also a vector.  Arithmetic operations on vectors in $\Z^B$ are defined as usual, extending the vectors with zeroes if necessary. For example, if $B' \subseteq B$, $u \in \Z^B$, and $v \in \Z^{B'}$, then $u + v \in \Z^B$ is defined by $(u + v)(b) = u(b) + v(b)$, where $v(b) = 0$ for every $b \in B\setminus B'$. For $u,v\in \Z^B$ we write $u \le v$ if $u_i\le v_i$ for all $i\in B$, and $u \lneqq v$ if $u\le v$ and $u\ne v$. Given a vector $\vec{v} \in \Z^k$, we define $\norm{\vec{v}}_1 = \sum_{i=1}^k \abs{\vec{v}_i}$ and $\norm{\vec{v}}_\infty = \max_{i=1}^k \abs{\vec{v}_i}$.

\subsection{Population protocols}
We recall the population protocol model of \cite{AngluinADFP06}, with explicit mention of leader agents.
A \emph{population protocol} is a tuple $\PP = (Q, \Tra, L, X, I, O)$ where 
\begin{itemize}
\item $Q$ is a finite set of \emph{states}; 
\item $\Tra \subseteq Q_2 \times Q_2$ is a set of \emph{transitions}, where $Q_2$ denotes the set of multisets over $q$ of size $2$;
\item $L \in \N^Q$ is the \emph{leader multiset}; 
\item $X$ is a finite set of \emph{input variables};  
\item $I \colon X \to Q$ is the \emph{input mapping}; and
\item $O \colon Q \to \{0, 1\}$ is the \emph{output mapping}.
\end{itemize}
\noindent  We write $p, q~\mapsto~p', q'$ to denote that the pair $(\multiset{p, q}, \multiset{p', q'})$ is a transition. We assume that for every multiset $\multiset{p, q}$ there is at least one transition of the form $p, q~\mapsto~p', q'$.


\smallskip\noindent\textbf{Inputs and configurations.} 
An \emph{input} to $\PP$ is a multiset $m \in \N^X$ such that $\size{m} \geq 2$. A \emph{configuration} is a 
multiset $\Conf \in \N^Q$ such that $\size{\Conf} \geq 2$. 
Intuitively, a configuration represents a population of agents where $\Conf(q)$ denotes the number of agents in
state $q$.  The \emph{initial configuration} for
input $\vec{m}$ is defined as  
$$\InC{\vec{m}} \defeq L \mplus \sum_{x \in X} \vec{m}(x) \cdot I(x) \ . $$ 
When $\PP$ has a unique input $x$, i.e.\ $X =\{x\}$, we abuse language and write $\InC{i}$ instead of  $\InC{i \cdot x}$ to denote the initial configuration for input $i\in\N$.

The \emph{output} $O(\Conf)$ of a
configuration $\Conf$ is $b$ if $\Conf(q) \geq 1$ implies $O(q) =b$ for all $q\in Q$,
and undefined otherwise.  So a population has output $b$ if all agents 
have output $b$. 

\smallskip\noindent\textbf{Executions.}  
A transition $t = \, p, q \mapsto p', q'$ is \emph{enabled} 
at a configuration $\Conf$ if $\Conf \geq p+q$, and \emph{disabled}
otherwise. As $\size{\Conf} \ge 2$ by the definition of configuration, every configuration enables at least one transition.
If $t$ is enabled at $\Conf$, then it can be \emph{fired}
leading to configuration $\Conff :=\Conf-p-q+p'+q'$,
which we denote $\Conf \trans{t} \Conff $.  We write $\Conf \trans{} \Conff $ if $\Conf \trans{t} \Conff $ for some $t \in T$. Given a sequence $\sigma=t_1 t_2 \ldots t_n$
of transitions, we write $\Conf \trans{\sigma} \Conff $ if there exist configurations
$\Conf_1, \Conf_2, \ldots, \Conf_n$ such that $\Conf \trans{t_1} \Conf_1 \trans{t_2} \Conf_2 \cdots \Conf_n  \trans{t_n} {\Conf}'$,
and $\Conf \transS \Conf'$ if $\Conf \trans{\sigma} {\Conf}'$ for some sequence $\sigma \in \Tra ^*$.
For every set of transitions ${\Tra}' \subseteq \Tra$, we
write $\Conf \trans{{\Tra}'} {\Conf}'$ if $\Conf \trans{t} {\Conf}'$ for some $t \in {\Tra}'$; we write  $\Conf \trans{{\Tra}'^*} {\Conf}'$, and say that $\Conf'$ is \emph{reachable} from $\Conf$,
if $\Conf \trans{\sigma} {\Conf}'$ for some sequence $\sigma \in {\Tra}'^*$.
Given a set $\Confset$ of configurations, $\Conf \transS \Confset$ denotes that 
$\Conf \transS {\Conf}'$ for some ${\Conf}' \in \Confset$.

An \emph{execution} is a sequence of configurations $\sigma = \Conf_0 \, \Conf_1 \, \ldots $
such that $\Conf_i \trans{} \Conf_{i+1}$ for every $i \in \N$. 
The \emph{output} $O(\sigma)$  of $\sigma$ is $b$ if there exist $i \in \N$
such that $O(\Conf_i) = O(\Conf_{i+1}) = ... = b$,  otherwise $O(\sigma)$
is undefined.
 
Executions have the \emph{monotonicity property}: If $\Conf_0 \, \Conf_1 \, \Conf_2 \ldots$
is an execution, then for every configuration $D$ the sequence $(\Conf_0 +\Conf) \, (\Conf_1 + \Conf) \, (\Conf_2 + \Conf) \ldots$
is an execution too. We often say that a statement  holds ``by monotonicity'' , meaning that it is 
a consequence of the monotonicity property.

\smallskip\noindent\textbf{Computations.} An execution $\sigma=\Conf_0 \, \Conf_1 \ldots$ is \emph{fair} if for every
configuration $\Conf$ the following holds: if $\Conf$ is reachable from $\Conf_i$ for infinitely many 
$i \in \N$, then $\Conf_j = \Conf$ for infinitely many $j \in \N$.
In other words, fairness
ensures that an execution cannot avoid a reachable configuration forever. We say
that a population protocol \emph{computes} a predicate $\varphi \colon
\N^X \to \{0, 1\}$ (or \emph{decides} the property represented by the predicate) if for every $\vec{v} \in \N^X$ every fair execution $\sigma$ starting from $\InC{\vec{v}}$ satisfies $O(\sigma) = \varphi(\vec{v})$. Two protocols are \emph{equivalent} if they
compute the same predicate. It is known that population protocols compute precisely the 
Presburger-definable predicates~\cite{AngluinAER07}. 

\begin{example}\label{ex:flock}
\newcommand{\Idx}{k}
  Let $\PP_\Idx = (Q, \Tra , \vec{0}, \{x\}, I, O)$ be the protocol where $Q
  \defeq \{0, 1, 2, 3, ..., 2^\Idx\}$, $I(x) \defeq 1$, $O(a) = 1$ if{}f $a = 2^\Idx$, and the set $\Tra $ of transitions contains $a,b \mapsto 0,a+b$  if $a+b<2^\Idx$, and  $a,b \mapsto 2^\Idx, 2^\Idx$ if $a
  + b \geq 2^\Idx$ for every $a, b \in Q$. It is readily seen that $\PP_\Idx$ computes $x \geq 2^\Idx$ with $2^\Idx+1$ states. Intuitively, each agent stores a number, initially 1. When two agents meet, one of them stores the sum of
  their values and the other one stores 0, with sums capping at
  $2^\Idx$. Once an agent reaches $2^\Idx$, all agents eventually get
  converted to~$2^\Idx$.

  Now, consider the protocol $\PP'_\Idx = (Q', \Tra ', \vec{0}, \{x\}, I',
  O')$, where $Q' \defeq \{0, 2^0, 2^1, ..., 2^\Idx\}$, $I'(x) \defeq
  2^0$, $O'(a) = 1$ if{}f $a = 2^\Idx$, and $\Tra '$ contains $2^i, 2^i \mapsto 0, 2^{i+1}$
  for each $0 \leq i < \Idx$, and $a, 2^\Idx \mapsto 2^\Idx, 2^\Idx$ for each $a \in Q'$.  It is easy to see
   that $\PP'_\Idx$ also computes $x \geq 2^\Idx$, but more succinctly; while $\PP_\Idx$ has $2^\Idx + 1$ states,
  $\PP'_\Idx$ has only $\Idx + 1$ states.
\end{example} 

\smallskip\noindent\textbf{Leaderless protocols.} A protocol $\PP = (Q, \Tra , L, X, I, O)$ has a multiset $L$ of leaders. If $L=0$, then the protocol is \emph{leaderless}. Protocols with leaders and leaderless protocols compute the same predicates \cite{AngluinAER07}. 
For $L=0$ we have 
$$\InC{\lambda \vec{v} + \lambda' \vec{v}'} = \lambda  \InC{\vec{v}} +  \lambda' \InC{\vec{v}'}$$
\noindent for all inputs $\vec{v}, \vec{v}' \in \N^X$ and  $\lambda, \lambda' \in \N$. In other words, any linear combination of initial configurations with natural coefficients is also an initial configuration. 


\subsection{State complexity of population protocols} 
Informally, the state complexity of a predicate is the minimal number of states of the protocols that compute it.
We would like to define the state complexity function as the function that assigns to a number $\ell$ the maximum state complexity of the predicates of size at most $\ell$. However, defining the size of a predicate requires to fix a representation. Population protocols compute exactly the predicates expressible in Presburger arithmetic \cite{AngluinAER07}, and so there are at least three natural representations: formulas of Presburger arithmetic, existential formulas of Presburger arithmetic, and semilinear sets \cite{Haase18}. Since the translations between these representations involve superexponential blow-ups,  we focus on threshold predicates of the form $x \geq \eta$, for which the size of the predicate is the size of $\eta$, independently of  
the representation. We choose to encode numbers in unary, and so we define $\STC(\eta)$ as the number of states of the smallest protocol computing $x \geq \eta$.

The inverse of $\STC(\eta)$ is the function that assigns to a number $n$ the largest $\eta$ such that a protocol with $n$ states computes $x \geq \eta$. Recall  that the busy beaver function assigns to a number $n$ the largest $\eta$ such that a Turing machine with $n$ states started on a blank tape writes $\eta$ consecutive ones on the tape and terminates. Due to this analogy, we call the inverse of the state complexity function the \emph{busy beaver function}, and call protocols computing predicates of the form $x \geq \eta$ busy beaver protocols, or just \emph{busy beavers}.

\begin{definition}
The \emph{busy beaver function} $\BB \colon \N \rightarrow \N$ is defined as follows: $\BB(n)$ is the largest $\eta \in \N$ such that the predicate $x \geq \eta$ is computed by some leaderless protocol with at most $n$ states. The function $\BBL(n)$ is defined analogously, but for general protocols, possibly with leaders.
\end{definition}

In \cite{BlondinEJ18} Blondin \textit{et al.} give lower bounds on the busy beaver function:

\begin{theorem}[\cite{BlondinEJ18}]
For each number of states $n$: $\BB(n) \in \Omega(2^n)$ and $\BBL(n) \in \Omega(2^{2^n})$.
\end{theorem}

However, to the best of our knowledge no upper bounds have been given. 

\section{Mathematical Structure of Stable Sets} \label{ssec:stablesets}
We define the \emph{stable configurations} of a protocol:

\begin{definition}
Let $b \in \{0,1\}$. A configuration $\Conf$ óf a protocol is \emph{$b$-stable} if $O(\Conf')=b$ for every configuration $\Conf'$ reachable from $\Conf$. The set of $b$-stable configurations is denoted $\SC_b$, and we let $\SC = \SC_0 \cup \SC_1$.
\end{definition}

It  follows easily from the definitions that a  population protocol \emph{computes} a predicate $\varphi \colon
\N^X \to \{0, 1\}$ if{}f for every input $\vec{v}$ and configuration $\Conf$ with $\InC{\vec{v}}\transS\Conf$ the condition $\Conf\transS\SC_{\varphi(v)}$ holds.

Moreover, given a protocol computing $\varphi$, $\varphi(v)=b$ for an input $v$ if{}f $\InC{\vec{v}}\transS\SC_{b}$.

A set $\Confset$ of configurations is \emph{downward closed} if $\Conf \in \Confset$ and $\Conff \le \Conf$ implies $\Conff \in \Confset$. The sets $\SC_0$, $\SC_1$, and $\SC$ are downward closed:

\begin{lemma}\label{lem:downwardsclosed}
Let $\PP$ be a protocol with $n$ states. For every $b \in \{0,1\}$ the set $\SC_b$ is downward closed.
\end{lemma}
\begin{proof}
Assume $\Conf \in \SC_b$ and $\Conf' \le \Conf$. We prove $\Conf' \in \SC_b$ by contradiction, so assume that $\Conf' \transS \Conf''$ for some $\Conf''$ such that $O(\Conf'')\neq b$. By monotonicity, $\Conf = \Conf'+ (\Conf-\Conf') \transS \Conf'' + (\Conf-\Conf')$, and since $O(\Conf'')\neq b$ we have $O(\Conf''+ (\Conf-\Conf')) \neq b$. So $\Conf' \in \SC_b$.
%
\end{proof}

Given a downward closed set $\Confset$,  a pair $(\Basel,S)$, where $\Basel$ is a configuration and $S \subseteq Q$,  is a \emph{basis element} of $\Confset$ if $\Basel+\N^S \subseteq \Confset$. A \emph{base} of $\Confset$ is a \underline{finite} set $\mathcal{B}$ of \emph{basis elements} such that $\Confset = \bigcup_{(\Basel,S) \in \mathcal{B}} ( \Basel+\N^S)$.  We define the \emph{norm} of a basis element $(\Basel,S)$ as $\norm{(\Basel, S)}_\infty := \norm{\Basel}_\infty$, and the norm of a basis as the maximal norm of its elements.

It is well-known that every downward-closed set of configurations has a base. We prove a stronger result: the sets $\SC_0$, $\SC_1$, and $\SC$ have bases of small norm.

\begin{lemma}\label{lem:basesaresmall}
Let $\PP$ be a protocol with $n$ states. Every $\Confset \in \{\SC_0,\SC_1,\SC\}$ has a basis of norm at most $2^{2(2n+1)!+1}$ with at most $\Basecount(n) := 2^{(2n+2)!}$ elements. 
\end{lemma}
\begin{proof}
For the bound on the norm, let $\beta:=2^{2(2n+1)!}$ and fix a $b$-stable configuration $\Conf$. Let $S:=\{q\in Q \mid \Conf(q)>2\beta\}$, and define $\Basel \leq \Conf$ as follows: $\Basel(i):=\Conf(i)$ for $i\notin S$ and $\Basel(i):=2\beta$ for $i\in S$. Since $\Basel \leq \Conf$ and $\Conf$ is $b$-stable, so is $\Basel$. We show that $(\Basel, S)$ is a basis element of $\SC_b$, which proves the result for $\SC_0$ and $\SC_1$.  Assume the contrary. Then some configuration $\Conf'\in\Basel+\N^S$ is not $b$-stable. So $\Conf' \transS \Conf''$ for some $\Conf''$ satisfying $\Conf''(q) \geq 1$ for some state $q\in Q$ with $O(q) \neq b$; we say that $\Conf''$ covers $q$.

By Rackoff's Theorem~\cite{rackoff78covering}, $\Conf''$ can be chosen so that $\Conf' \trans{\sigma} \Conf''$ for a sequence $\sigma$ of length $2^{2^{\O(n)}}$; a more precise bound is $\Cardinality{\sigma} \leq \beta$ (see Theorem~3.12.11 in \cite{esparza2019petri}). Since a transition moves at most two agents out of a given state, $\sigma$ moves at most $2\beta$ agents out of a state. So, by the definition of $\Basel$, the sequence $\sigma$ is also executable from $\Basel$, and also leads to a configuration that covers $q$. But this contradicts that $\Basel$ is $b$-stable. This concludes the proof for $\SC_0$ and $\SC_1$. For $\SC$, just observe that the union of the bases of $\SC_0$ and $\SC_1$ is a basis of $\SC$.

To prove the bound on the number of elements of the bases, observe that the number of pairs $(\Basel, S)$ 
such that $\Basel$ has norm at most $k$ and $S \subseteq Q$ is at most $(k+2)^n$. Indeed, for each state $q$ there are at most $k+2$ possibilities: $q \in S$, or $q \notin S$ and $0 \leq \Basel(q) \leq k$. So $\Basecount \leq  (2^{2(2n+1)!+1}+2)^n \leq  2^{(2n+2)!}$.
\end{proof}

From now on we use the following terminology:

\begin{definition}
\label{def:smallbasisconstant}
We call $\beta := 2^{2(2n+1)!+1}$ the \emph{small basis constant} for the protocol $\PP$.
A \emph{small basis} of $\SC_b$ or $\SC$ is a basis of norm at most $\beta$ (guaranteed to exist by Lemma \ref{lem:basesaresmall}). Its elements are called \emph{small basis elements}. 
\end{definition}


\section{A General Upper Bound on the Busy Beaver Function} \label{sec:generalbound}

We obtain a bound on the busy-beaver function $\BB_L(n)$.

Fix a protocol $\PP_n = (Q, \Tra , L, \{x\}, I, O)$ with $n$ states computing a predicate $x \geq \eta$. Observe that the unique input state is $x$, and so $\InC{a} = a \cdot x + L$ for every input $a$. 

Observe that for every input $i$ we have $\InC{i} \transS \Conf_i$ for some configuration $\Conf_i \in \SC$, and so $\Conf_i \in \Basel_i+ \N^{S_i}$ for some basis element $(\Basel_i, S_i)$ of $\SC$.  If $i < \eta$, then $\Conf_i \in \SC_0$, and if $i \geq \eta$ then $\Conf_i \in \SC_1$. Lemma \ref{lem:pumping} below uses this observation to provide a sufficient condition for an input $a$ to lie above $\eta$. The rest of the section shows that for a protocol with $n$ states some number $a < f(n)$ satisfies the condition, where $f(n)$ is a function from the Fast Growing Hierarchy \cite{FigueiraFSS11}. While the function $f(n)$ grows very fast, it is a recursive function. So, contrary to Turing machines, the busy-beaver function for population protocols does not grow faster than any recursive function.


\begin{lemma}
\label{lem:pumping}
If there exist $a,b \in \N$,  a basis element $(\Basel, S)$ of $\SC$, and configurations  $\Confa, \Confb \in \N^S$  satisfying 
\begin{enumerate}
\item $\InC{a} \transS \Basel + \Confa$,  and
\item $b \cdot x \transS \Confb$,
\end{enumerate}
\noindent then $\eta \leq a$.
\end{lemma}
\begin{proof}
We first claim that $\InC{a + \lambda b} \transS \Basel + \Confa + \lambda \Confb$ holds for every $\lambda \geq 0$.
Observe that 
$$\InC{a + \lambda b}  =  (a + \lambda b) \cdot x + L = \InC{a} + \lambda b \cdot x \ . $$ 
We have:
$$\begin{array}{rclcr}
\InC{a} + \lambda b \cdot x  & \trans{*} &  \Basel + \Confa + \lambda b \cdot x  &\quad \mbox{by (1)}  \\
 & \trans{*}&  \Basel + \Confa + \lambda \Confb &\quad \mbox{by (2)}  
\end{array}$$
\noindent and the claim is proved.

Assume now that $\eta > a$, i.e.\ $\PP_n$ rejects $a$. Since $\Confa \in \N^S$, we have $\Basel + \Confa \in \SC$, and so $\Basel + \Confa \in \SC_0$ because $\PP_n$ rejects $a$. Since $\Confb \in \N^S$, we have  $\Basel + \Confa + \lambda \Confb \in \Basel + \N^S$, and so $\Basel + \Confa + \lambda \Confb \in \SC_0$ for every $\lambda \geq 0$. So $\PP_n$ rejects $a + \lambda b$ for every $\lambda \geq 0$, contradicting that $\PP$ computes $x \geq \eta$.
\end{proof}

We now start our search for a number $a$ satisfying the conditions of Lemma \ref{lem:pumping}. First we identify a sequence of configurations $\Conf_2, \Conf_3, \Conf_4 \ldots$ of $\SC$ satisfying conditions close to 1. and 2. in Lemma \ref{lem:pumping}.

\begin{lemma}
\label{lem:sequence}
There exists a sequence $\Conf_2, \Conf_3, \Conf_4 \ldots$ of configurations of $\SC$ 
satisfying:
\begin{enumerate}[(1)]
\item \label{sequence1} $\InC{i} \transS \Conf_i$ for every $i \geq 2$, and
\item \label{sequence2} $\Conf_{i} + j \cdot x \transS \Conf_{i+j}$ for every $j \geq 0$. 
\end{enumerate}
\end{lemma}
\begin{proof}
Since $\PP_n$ computes $x \geq \eta$,  for every $i \geq 2$ every fair run of $\PP$ starting at $\InC{i}$ eventually reaches $\SC_0$ or $\SC_1$, depending on whether $i< \eta$ or $i \geq \eta$, and stays there forever. We define $\Conf_2, \Conf_3, \Conf_4, \ldots$ as follows. First, we let $\Conf_2$ be any configuration of $\SC$ reachable from $\InC{2}$. Then, for
every $i \geq 2$, assume that $\Conf_i$ has already been defined and satisfies $\InC{i} \transS \Conf_{i}$. Observe that 
$\InC{i+1} = \InC{i} \mplus x$. Since  $\InC{i} \transS \Conf_{i}$, we also have 
$\InC{i+1} = \InC{i} \mplus x \transS \Conf_i \mplus x$.  This execution can be extended to a fair run, which eventually reaches  $\SC$. We let $\Conf_{i+1}$ be any configuration of $\SC$ reachable from $\Conf_i \mplus x$. 

Let us show that $\Conf_2, \Conf_3, \Conf_4 \ldots$ satisfies \ref{sequence1} and \ref{sequence2}. Property \ref{sequence1} holds for $\Conf_2$ by definition, and for $i \geq 2$ because 
$\InC{i+1} = \InC{i} \mplus x \transS \Conf_i\mplus I(x) \transS \Conf_{i+1}$. For property \ref{sequence2}, by monotonicity and the definition of $\Conf_i$ we have for every $2 \leq i \leq k$:
\begin{multline*}
\Conf_{i} + j \cdot x  \transS  \Conf_{i+1} + (j-1) \cdot x \transS \\
\cdots  \transS   \Conf_{i+j-1} + x \transS \Conf_{i+j}
\end{multline*}
\end{proof}

We can now easily prove the existence of a number $a$ satisfying the conditions of Lemma \ref{lem:pumping}. We start by recalling Dickson's Lemma: 

\begin{lemma}[Dickson's lemma]
For every infinite sequence $v_1, v_2, \ldots $ of vectors of the same dimension there is an infinite sequence $i_1 < i_2 <  \ldots$ of indices such that $v_{i_1} \leq v_{i_2} \leq \ldots$.
\end{lemma}
\noindent 
By Dickson's lemma, the sequence $\Conf_2, \Conf_3, \Conf_4 \ldots $ of configurations of $\SC$ constructed in Lemma \ref{lem:sequence} contains an ordered subsequence $\Conf_{i_1} \leq \Conf_{i_2} \leq \Conf_{i_3} \cdots$. Since $\SC$ has a finite basis, by the pigeonhole principle there exist numbers $k < \ell$ and a basis element $(\Basel, S)$ such that $\Conf_{i_k} , \Conf_{i_\ell} \in \Basel + \N^S$. Since $\Conf_k \leq \Conf_\ell$, we have $\Conf_k - \Conf_\ell \in \N^S$, and so we can take:
$$a:=k;  b:= \ell - k; \Confa := \Conf_k - B ; \Confb := \Conf_\ell - \Conf_k \ . $$ 
However, the proof of Dickson's lemma is non-constructive, and gives no bound on the size of $a$. To solve this problem we observe that, in the terminology of \cite{FigueiraFSS11}, the sequence $\Conf_2 \, \Conf_3 \cdots $ is \emph{linearly controlled}:  there is a linear control function $f \colon \N \rightarrow \N$ satisfying $\Cardinality{\Conf_i} \leq f(i)$. Indeed, since  $\InC{i} \transS \Conf_i$, we have $\Cardinality{\Conf_i} = \Cardinality{\InC{i}}=\Cardinality{L} + i$, and so we can take $f(n) = \Cardinality{L} + n$.  This allows us to use a result on linearly controlled sequences from \cite{FigueiraFSS11}. Say a finite sequence $v_0, v_1, \cdots, v_s$ of vectors of the same dimension is \emph{good} if there are two indices $0 \leq i_1 < i_2 \leq s$ such that $v_{i_1} \leq v_{i_2}$.
The maximal length of good linearly controlled sequences has been studied in \cite{McAloon84,FigueiraFSS11,Balasubramanian20}.
In particular, this lemma  follows easily from results of \cite{FigueiraFSS11}:

\begin{lemma}{\cite{FigueiraFSS11}}
\label{lem:fastgrowing}
For every $\delta \in \N$ and for every elementary function $g \colon \N \rightarrow \N$,  there exists a function $F_{\delta,g} \colon \N \rightarrow \N$ at level $\mathcal{F}_{\omega}$ of the Fast Growing Hierarchy satisfying the following property: For every infinite sequence  $v_0, v_1, v_2 \ldots $  of vectors of $\N^n$ satisfying $\size{v_i} \leq i + \delta$, there exist $i_0 < i_1 < \ldots < i_{g(n)} \leq F_{\delta,g}(n)$ such that $v_{i_0} \leq v_{i_1} \leq \cdots \leq v_{g(n)}$.
\end{lemma}

We do not need the exact definition  of the Fast Growing Hierarchy (see \cite{FigueiraFSS11}); for our purposes it suffices to know that the level $\mathcal{F}_{\omega}$ contains functions that, crudely speaking, grow like the Ackermann function. From this lemma we obtain:

\begin{theorem}
\label{thm:generalbound}
Let $\PP_n$ be a population protocol with $n$ states and $\ell$ leaders computing a predicate $x \geq \eta$ for some $\eta \geq 2$. Then  $\eta < F_{\ell,\Basecount}(n)$, where $\Basecount(n)$ is the function of Lemma \ref{lem:basesaresmall}.
\end{theorem}
\begin{proof}
By Lemma \ref{lem:fastgrowing} there exist $\Basecount(n) +1$ indices
$i_0 < i_1 < \ldots < i_{\Basecount(n)} \leq F_{\ell, \Basecount}(n)$ such that  $\Conf_{i_0} \leq \Conf_{i_1} \leq \cdots \leq \Conf_{i_{\Basecount(n)}}$. By the definition of $\Basecount$ and the pigeonhole principle,
there are indices $k < \ell$ and a basis element $(\Basel, S)$ of $\SC_0$ such that  $\Conf_{i_k}, \Conf_{i_{\ell}} \in \Basel + \N^S$ and $\Conf_{i_k} \leq \Conf_{i_{\ell}}$.  By Lemma \ref{lem:pumping}, $\eta \leq i_k \leq F_{\ell,\Basecount}(n)$.
\end{proof}

\subsection{Is the bound optimal?}
The function $F_{\ell,\Basecount}(n)$ grows so fast that one can doubt that the bound is even remotely close to optimal. However, recent results show that this would be less strange than it seems. If a protocol $\PP$ computes a predicate $x \geq \eta$, then  $\eta $  is the smallest number such that $\InC{\eta} \transS \SC_1$. 
Therefore, letting $\BBP(n)$ denote the busy beaver protocols with at most $n$ states, and letting  $\SC_1^\PP$ and $\mathit{IC}^\PP$ denote the set $\SC_1$ and the initial mapping of the protocol $\PP$, we obtain:
\begin{equation*}
\BB_L(n)  =  \max _{\PP \in \BBP(n)} \min \{ i \in \N \mid \exists \Conf  \in  \SC_1^\PP  \colon \mathit{IC}^\PP\!(i) \transS \Conf \} 
\end{equation*}
\noindent Consider now a deceptively similar function. Let $\All_1$ be the set of configurations $\Conf$ such that $O(\Conf)=1$, i.e.\ all agents are in states with ouput $1$. 
Further, let $\Prot(n)$ denote the set of \emph{all} protocols with alphabet $X=\{x\}$, possibly with leaders, and $n$ states. Notice that we include also the protocols that do not compute any predicate. Define
\begin{equation*}
f(n)   =  \max _{\PP \in \Prot(n)} \min \{ i \in \N \mid \exists \Conf  \in  \All_1^\PP  \colon \mathit{IC}^\PP\!(i) \transS \Conf \} 
\end{equation*}

\noindent Using recent results in Petri nets and Vector Addition Systems \cite{CzerwinskiLLLM21,CO21,Leroux21,HornS20} it is easy to prove that $f(n)$ grows faster than any primitive recursive function\footnote{The paper \cite{HornS20} considers protocols with one leader, and studies the problem of moving from a configuration with the leader in a state $q_{in}$ and all other agents  in another state $r_{in}$, to a configuration with the leader in a state $q_f$ and all other agents in state $r_f$. Combined with \cite{CzerwinskiLLLM21,CO21,Leroux21}, this shows that the smallest number of agents  for which this is possible grows faster than any elementary function in the number of states of the protocol.}.
However,  a recent result \cite{BalasubramanianER21} by Balasubramanian \textit{et al.} shows $f(n) \in 2^{\O(n)}$ for leaderless protocols.

These results suggest that a non-elementary bound on $\BB_L(n)$ might well be optimal. However, in the rest of the paper we prove that this can only hold for population protocols with leaders. We show $\BB(n) \in 2^{2^{\O(n)}}$, i.e.\ leaderless busy beavers with $n$ states can only compute predicates $x \geq \eta$ for numbers $\eta$ at most double exponential in $n$.

\section{An Upper Bound for Leaderless Protocols}
\label{sec:leaderless bound}

Fix a \emph{leaderless} protocol $\PP_n = (Q, \Tra , \emptyset, \{x\}, I, O)$ with $\Abs{Q}=n$ states computing a predicate $x \geq \eta$. Observe that the unique input state is $x$, and so $\InC{a} = a \cdot x$ for every input $a$. We prove that $\eta\leq 2^{(2n+2)!} \in 2^{2^{\O(n)}}$.  We first introduce some well-known notions from the theory of Petri nets and Vector Addition Systems.

\subsection{Potentially realisable multisets of transitions}
\label{subsec:effective}

The \emph{displacement} of a transition $t= p, q \mapsto p', q'$ is the vector $\Displ{t} \in \{-2,-1, 0, 1, 2\}^Q$ given by $\Displ{t} :=p'+q'-p-q$. 
Intuitively, $\Displ{t}(q)$ is the change in the number of agents populating $q$ caused by the execution of $t$.
For example, if $Q = \{p, q, r\}$ and $t= p, q \mapsto p, r$ we have $\Displ{t}(p)=0$, $\Displ{t}(q)=-1$, and $\Displ{t}(r)=1$.
The \emph{displacement} of a multiset $\pi\in\N^T$ is defined as $\Displ{\pi} :=\sum_{t\in T}\pi(t) \cdot \Displ{t}$. 
We use the following notation:

\smallskip

\begin{center} $\Conf \potrans{\pi} \Conf'$ denotes that $\Conf' = \Conf + \Displ{\pi}$. \end{center}

\noindent Intuitively, $\Conf \potrans{\pi} \Conf'$ states that if $\Conf$ enables some sequence $t_1\, t_2 \ldots t_k \in T^*$ such that $\multiset{t_1, \ldots,t_k} = \pi$,
then the execution of $\sigma$ leads to $\Conf'$. However, such a sequence may not exist.  We call the multiset $\multiset{t_1, \ldots,t_k}$ the \emph{Parikh mapping} of  $t_1\, t_2 \ldots t_k$.

Say a configuration $C$ is \emph{$j$-saturated} if $C(q) \geq j$ for every $q \in Q$, i.e.\ if it populates all states with at least $j$ agents.
We have the following relations between $\trans{\sigma}$ and $\potrans{\pi}$:

\begin{lemma}
\label{lem:easy}
\begin{enumerate}[(i)]
\item \label{easy1} If $\Conf \trans{\sigma} \Conf'$ then $\Conf \potrans{\pi} \Conf'$, where $\pi$ is the Parikh mapping of $\sigma$
\item \label{easy2} If $\Conf \potrans{\pi} \Conf'$ and $\Conf$ is $2 \size{\pi}$-saturated, then $\Conf \trans{\sigma} \Conf'$ for any $\sigma$ with Parikh mapping $\pi$.
\end{enumerate}
\end{lemma}
\begin{proof}
\noindent\ref{easy1}: Easy induction on $\size{\sigma}$.\smallskip

\noindent\ref{easy2}:  By induction on $\size{\pi}$. The basis case $\pi=\emptyset$ is trivial. Otherwise, let $\sigma$ be any sequence with Parikh mapping $\pi$. Since this sequence is non empty, it can be decomposed as $t\sigma'$ where $t$ is a transition and $\sigma'$ is a sequence with Parikh mapping $\pi'$ defined by $\pi'(t) = \pi(t) - 1$ and $\pi'(t') = \pi(t')$ for every $t' \neq t$. As $\Conf$ is $2\size{\pi}$-saturated, we have $\Conf \trans{t} \Conf''$ for some configuration $\Conf''$. Further,  $\Conf''$ is $2 \size{\pi'}$-saturated and $\Conf'' \potrans{\pi'} \Conf'$. By induction hypothesis 
$\Conf'' \trans{\sigma'} \Conf'$. It follows that $\Conf\trans{\sigma} \Conf'$, and we are done.
\end{proof}

\noindent We introduce the set of potentially realisable multisets of transitions of a protocol:

\begin{definition}\label{def:potentiallyrealisable}
A multiset $\pi$ of transitions is \emph{potentially realisable} if there are $i \in \N$ and $\Conf \in \N^Q$ such that $\InC{i}\potrans{\pi}\Conf$. 
\end{definition}
\noindent The reason for the name is as follows. If $\pi$ is not potentially realisable, then by Lemma \ref{lem:easy}\ref{easy1} no sequence $\sigma \in T^*$ with Parikh mapping $\pi$ can be executed from any initial configuration, and so $\pi$ cannot be ``realised''. In other words,  potential realisability is a necessary but not sufficient condition for the existence of an execution that ``realises'' $\pi$. 

\subsection{Structure of the proof}

We  can now give a high-level view of the proof of the bound $\eta\leq 2^{(2n+2)!}$. The starting point is a version of Lemma \ref{lem:pumping} in which,
crucially, the condition $\InC{b} \transS \Confb$ is replaced by the weaker $\InC{b} \potransS \Confb$.

\begin{lemma}
\label{lem:pumping2}
If there exist $a,b \in \N$,  a basis element $(\Basel, S)$ of $\SC$, configurations  $\Confa, \Confb \in \N^S$,
and a configuration $\LargeC$ satisfying 
\begin{enumerate}[(i)]
\item \label{pumping1} $\InC{a} \transS \LargeC \transS \Basel + \Confa$, and 
\item \label{pumping2} $\InC{b} \potrans{\pi} \Confb$ for some $\pi \in \N^T$ such that $\LargeC$ is $2 \size{\pi}$-saturated,
\end{enumerate}
\noindent then $\eta \leq a$.
\end{lemma}
\begin{proof}
We first claim that
\begin{equation*}
\label{claim}
\InC{a + \lambda b} \transS B+ \Confa + \lambda \Confb
\end{equation*}
holds for every $\lambda \geq 0$. To prove this, observe first that
$\InC{b} \potrans{\pi} \Confb$ implies $\LargeC + \InC{b} \potrans{\pi} \LargeC + \Confb$. Since
$\LargeC$ is $2 \size{\pi}$-saturated so is $\LargeC + \InC{b}$, and, by Lemma \ref{lem:easy}\ref{easy2}, we have $\LargeC + \InC{b} \trans{\sigma} \LargeC + \Confb$ where $\sigma$ is any sequence with Parikh mapping $\pi$. With an induction on $\lambda$, we immediately derive
\begin{equation}
\label{largereach}
\LargeC + \lambda\InC{b} \trans{\sigma^\lambda} \LargeC + \lambda\Confb \ . \tag{$*$}
\end{equation}
\noindent 
Since $\PP_n$ is leaderless,  $\InC{a + \lambda b}  =  \InC{a} + \lambda \InC{b}$ holds, and:
$$\begin{array}{rclll}
\InC{a} + \lambda \InC{b} & \transS &  \LargeC  + \lambda \InC{b}  &\quad & \mbox{by \ref{pumping1}}  \\
&  \transS &  \LargeC  + \lambda \Confb & \quad & \mbox{by (\ref{largereach})}  \\
& \transS  & \Basel + \Confa + \lambda \Confb  & \quad & \mbox{by \ref{pumping1}.}
\end{array}$$
\noindent This proves the claim.

Assume now that $\eta > a$, i.e.\ $\PP_n$ rejects $a$. Since $\Confa,\Confb \in \N^S$, we have $B+\Confa+\lambda\Confb \in \SC$, and so $B+\Confa+\lambda\Confb \in \SC_0$ for every $\lambda \geq 0$. So, by the claim, $\PP_n$ rejects $a + \lambda b$ for every $\lambda \geq 0$, contradicting that it computes $x \geq \eta$.
\end{proof}

In the next sections we show that $a:= 2^{(2n+2)!}$ satisfies the conditions of Lemma \ref{lem:pumping2}. We proceed in three steps:

\begin{enumerate}[(a)]
\item \label{item:step1} Section \ref{subsec:saturated} proves that for every $j \in \N$ and input $a \geq j 3^n$ the initial configuration
$\InC{a}$ can reach a $j$-saturated configuration $\LargeC$. We remark that this result is only true for leaderless protocols.
\item \label{item:step2} 
Let $S\subseteq Q$. Section \ref{subsec:epsconcentrated} proves that $\InC{a}\potransS B+\Confa$, for $a\in\N$, $\Confa\in\N^S$ and a configuration $B$, implies $\InC{b}\potransS\Confb$ for some $b\in\N,\Confb\in\N^S$, provided that $a$ is “large” relative to $\abs{B}$.
\item \label{item:step3} Section \ref{subsec:bound} puts everything together, and gives the final bound.
\end{enumerate}

\subsection{Reaching $j$-saturated configurations} \label{subsec:saturated}

Recall our assumption that for every state $q \in Q$ there exists an input $i_q$ such that $\InC{i_q} \transS C_q$ for some configuration $C_q$ such that $C_q(q)> 0$.
By monotonicity, we have $\InC{i} \transS C$ for the input $i := \sum_{q \in q} i_q$ and the $1$-saturated configuration $C := \sum_{q \in Q} C_q$. 
We show that we can choose $i < 3^n$ and that $\InC{i} \trans{\sigma} C$ for some $\sigma$ such that $\size{\sigma} \leq 3^n$. 
It follows that for every $j \in \N$ the input $j 3^n$ can reach a $j$-saturated configuration by executing $j$ times $\sigma$.

\begin{lemma}\label{lem:findtransition}
 Let $\Conf$ be a configuration satisfying $x\in \supp{\Conf} \subset  Q$. There exists a transition $p,q \mapsto p', q'$ such that $\{p,q\} \subseteq \supp{\Conf}$ and $\{p', q' \} \not\subseteq \supp{\Conf}$.
\end{lemma}
\begin{proof}
Since $\supp{\Conf}$ is strictly included in $Q$, there exist $i\in\setN$, a word $\sigma\in T^*$, a configuration $\Conf'$ such that $\InC{i} \trans{\sigma}\Conf'$ and $\supp{\Conf'} \not\subseteq \supp{\Conf}$. Assume w.l.o.g. that $\sigma$ has minimal length. If $\sigma=\epsilon$ then $\supp{\Conf'}\subseteq \{x\}$ contradicting $\supp{\Conf'} \not\subseteq \supp{\Conf}$. So $\sigma=\sigma't$ for some transition $t=p,q \mapsto p', q'$, and  $\InC{i} \trans{\sigma'}\Conf'' \trans{t} \Conf'$ for some configuration $\Conf''$. From $\Conf'' \trans{t} \Conf'$ we derive $\{p,q\}\subseteq \supp{\Conf''}$ and $\supp{\Conf'}\subseteq\supp{\Conf''}\cup\{p',q'\}$. By minimality of $\sigma$ and $\size{\sigma'}<\size{\sigma}$, we deduce that $\supp{\Conf''}\subseteq \supp{\Conf}$. In particular $\{p,q\}\subseteq \supp{\Conf}$, and $\supp{\Conf'}\subseteq \supp{\Conf}\cup\{p',q'\}$. As $\supp{\Conf'} \not\subseteq \supp{\Conf}$, we deduce that $\{p',q'\}\not\subseteq \supp{\Conf}$.
\end{proof}

\begin{lemma}\label{lem:allpopulated}
There exists a word $\sigma\in T^*$, a $1$-saturated configuration $\Conf$ and a sequence $\sigma$ of length at most $3^n$ such that  $\InC{3^n} \xrightarrow{\sigma} \Conf$.
\end{lemma}
\begin{proof}
We build a sequence $\sigma_0, \sigma_1 \ldots$ of words in $T^*$ and a non decreasing sequence $\Conf_0, \Conf_1 \ldots$ of configurations such that for every $j \geq 0$ we have (with the convention $\Conf_{-1}=0$)
  \begin{description}
  \item[$\bullet$] $\InC{3^j} \trans{\sigma_j} \Conf_j$;
  \item[$\bullet$] $\supp{\Conf_{j-1}}$ is saturated, or $\supp{\Conf_{j-1}} \subset \supp{\Conf_j}$, and
  \item[$\bullet$] $\Cardinality{\sigma_j}=(3^j-1)/2$.
  \end{description}
\noindent Since $\PP_n$ has $n$ states, some $\Conf_j$ with $0\leq j \leq n$ is saturated. 

The sequence is built inductively. Choose $\Conf_0 = \InC{1}$ and $\sigma_0 = \epsilon$. Assume that
$\sigma_0, \ldots, \sigma_k$ has been built. If $\Conf_k$ is saturated, then choose $\Conf_{k+1} := 3\Conf_k$ and $\sigma_{k+1} := \sigma_k^3$.
Assume $\Conf_k$ is not saturated. Since $x\in  \supp{\Conf_1}$ we have $x\in  \supp{\Conf_k} \subset Q$. By  Lemma~\ref{lem:findtransition} there exists a transition $t= p, q \mapsto p', q'$ such that $p,q\in \supp{\Conf_k}$ and $\{p',q'\}\not\subseteq  \supp{\Conf_k}$. Since $\InC{3^k} \trans{\sigma_k}\Conf_k$ and $\PP_n$ is leaderless, we have $\InC{3^{k+1}} \trans{\sigma_k^3} 3\Conf_k$. 
Since $p,q\in \supp{\Conf_k}$, transition $t$ is enabled at $2\Conf_k$. We choose $\Conf_{k+1}:=\Conf_k+\Conf_k'$ where $\Conf_k'$ is  the configuration satisfying $2\Conf_k\trans{t}\Conf_k'$. Now, just set $\sigma_{k+1}:=\sigma_k^3 t$.
\end{proof}

\subsection{Reaching $\epsilon$-concentrated stable configurations}
\label{subsec:epsconcentrated}

This section contains the core of the proof. We want to find some $b\in\N$ and a configuration $\Confb\in\N^S$ with $\InC{b}\potransS\Confb$, where $S\subseteq Q$ is given by a basis element $(\Basel,S)$. We start by noting that any sequence $\InC{a}\transS B+\Confa$, again with $a\in\N$ and $\Confa\in\N^S$, already fulfils this condition \emph{approximately} if $a$ is large and $(\Basel,S)$ is a small basis element, i.e.\ $B$ is “small” relative to $\Confa$. The next lemma formalises this notion.

\begin{definition}\label{def:epsconc}
Let $0 < \epsilon \leq 1$ and $S \subseteq Q$. A configuration $C$ is \emph{$\epsilon$-concentrated in $S$} if $C(S) \geq (1 -\epsilon) \size{C}$ or, equivalently, $C(Q\setminus S) \leq \epsilon \size{C}$. 
\end{definition}

\newcommand{\loww}{k n \beta}
\begin{lemma}\label{lem:big}
Let $k\geq 1$, let $a:= \loww$, where $\beta$ is the constant of Definition \ref{def:smallbasisconstant}, and let $\LargeC$ denote a configuration with $\InC{a}\transS\LargeC$. There exists a small basis element $(\Basel,S)$ of $\SC$, and a $\Confa \in \setN^S$ s.t.
\begin{enumerate}
\item\label{lem:big:ca} $\InC{a} \transS \LargeC \transS \Basel + \Confa$, and
\item\label{lem:big:cc} $\Basel + \Confa$ is $\frac{1}{k}$-concentrated in $S$.
\end{enumerate}
\end{lemma}
\begin{proof}
As $\LargeC$ is reachable from an input configuration, it can reach a configuration $\Basel+\Confa$, where $\Confa\in\N^S$ and $(\Basel,S)$ is a small basis element of $\SC$. The latter implies $\norm{\Basel}_\infty\le\beta$, so we have $\abs{\Basel}\le n\beta=a/k=\abs{\Basel+\Confa}/k$, and $\Basel+\Confa$ is $\tfrac{1}{k}$-concentrated in $S$.
\end{proof}

In the remainder of this section we use another small basis theorem, Pottier's small basis theorem for Diophantine equations, to extend the above result: not only can we reach $\epsilon$-concentrated configurations for arbitrarily small $\epsilon>0$, we can also \emph{potentially} reach a $0$-concentrated configuration, i.e.\ a configuration in which \emph{all} agents populate $S$.

For this, we observe that the potentially realisable multisets $\pi$ (see Definition~\ref{def:potentiallyrealisable}) are precisely the solutions of the system of $\size{Q}-1$ Diophantine equations over the variables 
$\{ \pi(t) \}_{t \in T}$ given by \[\sum_{t \in T} \pi(t) \Displ{t}(q) \ge0 \qquad\mbox{ for $q \in Q\setminus\{x\}$.}\]

Let $A \cdot y \geq 0$ be a homogeneous system of linear Diophantine inequalities, i.e.\ a solution is a vector $m$ over the natural numbers such that $A \cdot m \geq 0$. A set $\Pottbase$ of solutions is a basis
if every solution is  the sum of a multiset of solutions of $\Pottbase$. Formally, $\Pottbase$ is a \emph{basis} if for every solution $m$, there exists a multiset $M \in \N^\Pottbase$ such that $m = \sum_{b\in \Pottbase}M(b)\cdot b$.
It is easy to see that every system has a finite basis. Pottier's theorem shows that it has a small basis:

\begin{theorem}[\cite{Pottier91}]\label{thm:pottier}
Let $A \cdot y \geq 0$ be a system of $e$ linear Diophantine equations on $v$ variables.
There exists a basis $\Pottbase \subseteq \N^v$ of solutions such that for every $m \in \Pottbase$:
$$\norm{m}_1 \leq \bigg(  1 + \max_{i=1}^e \sum_{j=1}^v \abs{a_{ij}}  \bigg)^e \ .$$
\end{theorem}
Since the potentially realisable multisets are the solutions of a system of Diophantine equations, by applying Pottier's theorem we obtain:
\begin{corollary}\label{cor:pottier}
Let $\pottname := \pottcons$. There exists a basis of potentially realisable multisets such that every element $\pi$ of the basis satisfies  $\size{\pi} \leq \pottname/2$. Moreover, $\InC{i} \potrans{\pi} \Conf$ for some $i \leq \pottname$  and some $\Conf$ such that  $\Conf(Q) \leq \pottname$ and $\Conf(x)=0$.
\end{corollary}
\begin{proof}
The potentially realisable multisets are precisely the solutions of a system of $\size{Q}-1$ Diophantine equations on $\size{T}$ variables.
The matrix $A$ of the system satisfies $\abs{a_{ij}} = \abs{\Displ{j}(q_i)}$, and so $\abs{a_{ij}} \leq 2$ for every $1 \leq i<\size{Q}$ and every $1 \leq j \leq \size{T}$.
The result follows from Theorem \ref{thm:pottier}.

For $\Conf(x)=0$, note that $\InC{i} \potrans{\pi} \Conf$ implies $\InC{i-\Conf(x)} \potrans{\pi} \Conf-x\cdot\Conf(x)$.
\end{proof}

Recall that Definition \ref{def:smallbasisconstant} introduced the small basis constant $\beta$. Analogously, we now introduce the Pottier constant:

\begin{definition}
\label{def:pottierconstant}
We call $\pottname := \pottcons$ the \emph{Pottier constant} for the protocol $\PP$. 
\end{definition}

\begin{remark}
For a deterministic population protocol (i.e.\ a protocol such that for every pair of states there is at most one transition) we could instead use $\pottname=2(\abs{Q}+2)^{\abs{Q}}$.
\end{remark}

The following lemma is at the core of our result. Intuitively, 
it states that if some potentially realisable multiset leads to a $(1 / \pottname)$-concentrated configuration in $S$, then some  \emph{small} potentially realisable multiset leads to a $0$-concentrated configuration in $S$.

\begin{lemma}\label{lem:allempty}
Let $\InC{i} \potransS \Conf$ and $S \subseteq Q$. If $C$ is $(1/\pottname)$-concentrated in $S$, then $\InC{j} \potrans{\theta} \Conf'$ for some number $j$, potentially realisable multiset $\theta \in \N^T$, and configuration $\Conf'$ such that $\size{\theta} \leq \pottname/2$, $j \leq \pottname$ and $\Conf'$ is $0$-concentrated on $S$, i.e.\ $\Conf' \in \N^S$.
\end{lemma}

\begin{proof}
\newcommand{\Xcount}{l}
If $x\notin S$ then we take $\theta$ as the empty multiset and are
done. Otherwise let $\Xcount:=\Conf(x)$ and set $i^*:=i-\Xcount$ and
$C^*:=C-\Xcount\cdot x$. Clearly, $\Conf^*(S)< i^*/\pottname$ follows
from $\Conf(S)<i/\pottname$. It therefore suffices to prove the lemma in
the case of $\Conf^*=\Conf$ and $i^*=i$, i.e.\ $\Xcount=0$.

Since $\InC{i} \potransS \Conf$, we have $\InC{i} \potrans{\pi} \Conf$
for some potentially realisable multiset $\pi$. By Corollary \ref{cor:pottier}
there exist potentially realisable multisets $\pi_1, \ldots, \pi_k \in \N^T$ s.t.\
$\pi = \pi_1+...+\pi_k$, numbers $i_1, \ldots, i_k \in \N$ and
configurations $\Conf_1, \ldots, \Conf_k$ (not necessarily distinct)
such that  
\begin{description}
\item[$\bullet$] $i = \sum_{j=1}^k i_j$,  $\pi = \sum_{j=1}^k \pi_j$, and $\Conf =
\sum_{j=1}^k \Conf_j$.
\item[$\bullet$] $\InC{i_j} \potrans{\pi_j} \Conf_j$.
\item[$\bullet$] $i_j \leq \pottname$ and $\Conf_j(Q) \leq \pottname$ for
$j\in\{1,..,k\}$.
\end{description}
Let  $J=\{j\in\{1,\ldots,k\} \mid i_j > 0\}$. We have
$$i = \sum_{j=1}^k i_j = \sum_{j\in J} i_j \leq \pottname\cdot\abs{J} \ .$$
Since $i > \pottname \cdot \Conf(S)$ by hypothesis, we deduce that
$\abs{J}>\Conf(S)$.  Assume that $\Conf_j(S)>0$ for every $j\in J$.
Then $\Conf(S)=\sum_{j\in J}\Conf_j(S)\geq \abs{J}$, a contradiction.
So $\Conf_j(S)=0$ for some $j\in J$, and we take $\theta :=\pi_j$.
\end{proof}

\subsection{The upper bound}
\label{subsec:bound}

We have now obtained all the results needed for our final theorem.

\begin{theorem}
Let $\PP_n$ be a leaderless population protocol with $n$ states,
and let $\beta$ and $\pottname$ be the constants of Definition \ref{def:smallbasisconstant} and Definition \ref{def:pottierconstant}.
If  $\PP_n$ computes a predicate $x \geq \eta$, then 
$$\eta \le \pottname n \beta 3^n \leq 2^{(2n+2)!}  \ .$$
\end{theorem}
\begin{proof}

Assume $\PP_n$ computes $x \geq \eta$. We first prove that $\eta \le a$ holds for $a:= \pottname n \beta 3^n$. It suffices to show that $a$ satisfies the conditions of Lemma \ref{lem:pumping2} for a suitable choice 
of $b$, $(\Basel, S)$, $\Confa$, $\Confb$, $\LargeC$, and $\pi$.

Lemma~\ref{lem:allpopulated} shows that a $1$-saturated configuration $\Conf$ with $\InC{3^n}\rightarrow\Conf$ exists. Choose $\LargeC:=\pottname n \beta\cdot\Conf$. Observe that $\LargeC$ is $(\pottname n \beta)$-saturated and reachable from $\InC{a}$.

Now choose  $(\Basel, S)$ and $\Confa$ as the configurations and basis element of Lemma \ref{lem:big}, where $k:=\pottname3^n$ and $\LargeC$ is chosen as above. With this choice of $k$ Lemma \ref{lem:big} yields:
\begin{enumerate}
\item\label{lem:big2:ca} $\InC{a} \transS \LargeC \transS \Basel + \Confa$; and
\item\label{lem:big2:cc} $\Basel + \Confa$ is $(1/\pottname)$-concentrated in $S$.
\end{enumerate}
(Regarding \ref{lem:big2:cc}., note that $(1/\pottname3^n)$-concentrated implies $(1/\pottname)$-concentrated.)
Property \ref{lem:big2:ca}.\ coincides with condition \ref{pumping1} of Lemma~\ref{lem:pumping2}. 
Finally, we choose $b$, $\Confb\in\N^S$, and $\pi$ so that condition \ref{pumping2} of Lemma~\ref{lem:pumping2} holds as well. In particular, $b$, $\Confb$, and $\pi$ must satisfy:
\begin{center}$\InC{b} \potrans{\pi} \Confb$ for some $\pi$ s.t.\ $\LargeC$ is $2 \size{\pi}$-saturated.\end{center}
\noindent By \ref{lem:big2:ca}.\ we have $\InC{a} \potransS \Basel + \Confa$.
Applying Lemma \ref{lem:allempty} with $i:=a$ and $\Conf := \Basel + \Confa$, we conclude that $\InC{j} \potrans{\theta} \Conf'$ for some $0<j \leq \pottname$, some multiset $\theta$ such that $\size{\theta} \leq \pottname/2$, and some configuration $\Conf'$ $0$-concentrated in $S$, i.e.\ satisfying $C'(Q\setminus S)=0$.

Using $b := j$, $\pi := \theta$, and $\Confb := \Conf'$, we have $\InC{b} \potrans{\pi} \Confb$, which fulfils condition \ref{pumping2} of Lemma~\ref{lem:pumping2} as well. Further, since
$\LargeC$ is $(n \beta \pottname)$-saturated and $n \beta \pottname \geq 2 (\pottname/2) \geq 2 \size{\pi}$, the configuration $\LargeC$ is $2\size{\pi}$-saturated. Applying the lemma we get $\eta \le a$, concluding the first part of our proof.

We now show that $\pottname n \beta 3^n \leq 2^{(2n+2)!}$. Since $\pottname = \pottcons$ and $\beta = 2^{2(2n+1)!+1}$, we
deduce (for $n\ge2$):
\begin{align*}
\eta&\le2\,(2\abs{T}+1)^n\,n\,2^{2(2n+1)!+1}\,3^n\\
&\le4\,(2n^4+1)^n\,n\,2^{2(2n+1)!}\,2^{2n}\\
&\le4\cdot3^nn^{4n}n2^{2(2n+1)!}\,2^{2n}\\
&\le4\cdot3^nn^{4n+1}2^{2(2n+1)!}\,2^{2n}
\end{align*}
\noindent and so
\begin{align*}
2^\eta&\le2+2n+(4n+1)\log_2n+2(2n+1)!+2n\\
&\le2(2n+1)!+(2n+1)(2+2\log_2n)\\
&\le(2n+2)!
\end{align*}
Therefore $\eta\le 2^{(2n+2)!}$.
\end{proof}





\section{Conclusion} \label{sec:conclusion}
We have obtained the first non-trivial lower bounds on the state complexity of population protocols computing counting predicates of the form $x \geq \eta$, a fundamental question about the model. The obvious open problems are to 
close the gap between 
the $\Omega(\log \log \eta)$ lower bound and the $\O(\log \eta)$ upper bound for the leaderless case, and the even larger gap between $\O(\log \log \eta)$ and (roughly speaking), the $\Omega(\alpha(\eta))$ lower bound for protocols with leaders, where $\alpha(\eta)$ is the inverse of the Ackermann function.

\bibliographystyle{plainurl}
\bibliography{references}

\begin{appendices}

\label{sec:appendix}

\end{appendices}

\end{document}